\documentclass[12pt]{article}
\usepackage{graphicx,amsmath,amsfonts,amssymb,bbm}
\usepackage{psfrag}
\usepackage{epsf}
\newtheorem{theorem}{Theorem}
\newtheorem{lemma}{Lemma}
\newtheorem{remark}{Remark}
\newenvironment{proof}
{\begin{trivlist} \item[]{\bf Proof. }}%
{\hspace*{\fill}$\rule{.3\baselineskip}{.35\baselineskip}$\end{trivlist}}
\begin{document}
\title{\bf \large Persistence of the Thomas-Fermi approximation for ground states supported
by the nonlinear confinement}

\author{Boris A. Malomed$^{1}$ and Dmitry E. Pelinovsky$^{2,3}$ \\
{\small \it $^{1}$ Department of Physical Electronics, Tel Aviv University, Tel Aviv 69978, Israel} \\
{\small \it $^{2}$ Department of Mathematics, McMaster
University, Hamilton, Ontario, L8S 4K1, Canada} \\
{\small \it $^{3}$ Department of Applied Mathematics, Nizhny Novgorod State
Technical University, Russia} }

\maketitle

\begin{abstract}
We justify the Thomas--Fermi approximation for the elliptic problem
with the repulsive nonlinear confinement used in the recent physical literature.
The method is based on the resolvent estimates and the fixed-point iterations.
\end{abstract}

Self-trapping of solitary waves in nonlinear physical
media is a commonly known problem of profound significance \cite{KA,DP}. An
obvious condition is that attractive (alias self-focusing) nonlinearity is
necessary for the creation of localized states. Recently, a radically
different approach to this problem was proposed in Refs. \cite{Barcelona1,Barcelona2,Barcelona-2comp}:
\emph{repulsive} (self-defocusing) nonlinearity that grows at infinity
readily gives rise to the self-trapping of localized states,
which are stable to weak and strong perturbations alike.
%The necessary
%spatially inhomogeneous nonlinearity profile in Bose-Einstein condensates
%may be induced by means of the Feshbach resonance, controlled by magnetic
%\cite{magn-Feshbach,magn-Feshbach2} or optical \cite%
%{optical-Feshbach,optical-Feshbach2} fields. Appropriate magnetic-field
%landscapes can be created by magnetic lattices \cite%
%{magnetic-lattice,magnetic-lattice2}, while the necessary profiles of the
%optical intensity may be effectively painted by laser beams \cite{paint}. In
%photonics, nonuniform distribution of the dopant density can be used to
%build similar profiles of the local strength of the self-defocusing
%nonlinearity \cite{Kip}.

An advantage offered by models with the effective \textit{nonlinear
confinement} is a possibility to find particular solutions in an
exact form \cite{Barcelona3,Barcelona4}, and to apply analytical methods to the qualitative approximation
of various localized states \cite{Barcelona5}. The simplest method used for approximating
the ground state of energy is the Thomas-Fermi (TF) approximation \cite{Thomas,Fermi}.
Comparison with numerical results has
demonstrated that the TF approximation produces quite accurate results for
the self-trapped modes with sufficiently large amplitudes \cite{Barcelona1,Barcelona5}. The
objective of the present work is to produce a rigorous estimate of
the proximity of the TF approximation to true ground states in models
with the spatially growing strength of the defocusing cubic nonlinearity.

In a similar context of the elliptic problem with
the harmonic confinement and the defocusing cubic nonlinearity, the TF approximation
was rigorously justified using calculus of variations \cite{IM} and
reductions to the Painlev\'e-II equation \cite{GP,KS}. The difficulty that
arises in this context is that the Thomas--Fermi approximation is
compactly supported and the derivatives of the ground state diverge in a transitional
layer near the boundary. Compared to this complication, we show that the
justification of the Thomas--Fermi approximation
in the elliptic problem with the nonlinear confinement can be obtained from the standard resolvent
estimates and fixed-point arguments.

Following to the main model used in Refs. \cite{Barcelona1}--\cite{Barcelona5},
we consider the elliptic problem with the repulsive nonlinear confinement,
\begin{equation}
-\epsilon ^{2}\Delta u+V(x)u^{3}-u=0,\quad x\in \mathbb{R}^{d}, \quad
d = 1,2,3, \label{1}
\end{equation}%
where $\epsilon $ is a small parameter corresponding to the TF
approximation, $\Delta $ is the $d$-dimensional Laplacian,
$u$ is a positive stationary state to be
found, and, in accordance with what is said above, the strength of the nonlinear confinement
$V$ is supposed to satisfy the
following properties: (i) $V(x) \geq V_0 > 0$ for all $x\in \mathbb{R}^{d}$, and (ii) $%
V(x)\rightarrow \infty $ as $|x|\rightarrow \infty $. Further constraints on
the smoothness of $V$ and its growth at infinity will be needed for the main
result. Note, however, that no symmetry assumptions are needed on $V$.

The formal TF solution of the elliptic problem is found for $\epsilon =0$
and corresponds to the spatially decaying positive eigenfunction:
\begin{equation}
u_{0}(x)=\frac{1}{\sqrt{V(x)}},\quad x\in \mathbb{R}^{d}. \label{TF-approximation}
\end{equation}%
If we require $u_{0}\in L^{2}(\mathbb{R}^{d})$ so that
the stationary state can be normalized in the $L^2(\mathbb{R}^d)$ norm,
then $1/V$ needs to be integrable. However, this requirement is not needed
for the main persistence result formulated as follows.

\begin{theorem}
\label{theorem-fixed-point} Assume that $\nabla \log (V) \in H^{2}(\mathbb{R}^{d})$
for $d = 1$ or $\nabla \log (V) \in H^{3}(\mathbb{R}^{d})$ for $d = 2,3$.
There exist positive constants $\epsilon _{0}$ and $C_{0}$ such that for
every $\epsilon \in (0,\epsilon _{0})$, there exists a unique solution $u = u_0 + U$
of the nonlinear elliptic problem (\ref{1}) with
$U \in H^{1}(\mathbb{R}^{d})$ satisfying
\begin{equation}
\Vert U \Vert _{H^{1}}\leq C_{0}\epsilon ^{2}.
\end{equation}
\end{theorem}

To study the persistence of the TF approximation, we set $u:=w/\sqrt{V}$ and
decompose $w:=1+r$. In this way, the nonlinear elliptic problem (\ref{1})
can be rewritten for the perturbation function $r$:
\begin{equation}
L_{\epsilon }r=\epsilon ^{2}F+N(r),  \label{persistence}
\end{equation}%
where $N(r)=-3r^{2}-r^{3}$ is the nonlinear term,
\begin{equation}
F=\sqrt{V}\Delta \frac{1}{\sqrt{V}}=-\frac{\Delta V}{2V}+\frac{3|\nabla
V|^{2}}{4V^{2}}
\end{equation}%
is the source term, and
\begin{equation}
L_{\epsilon }=2-\epsilon ^{2}\Delta +\epsilon ^{2}\frac{1}{V}\nabla V\cdot
\nabla -\epsilon ^{2}F
\end{equation}%
is the linearized operator at the TF approximation. Further, we write $%
L_{\epsilon }$ as a sum of two operators,
\begin{equation}
\tilde{L}_{\epsilon }:=2-\epsilon ^{2}\Delta -\frac{\epsilon ^{2}|\nabla
V|^{2}}{4V^{2}}
\end{equation}%
and
\begin{equation}
L_{\epsilon }-\tilde{L}_{\epsilon }:=\epsilon ^{2}\frac{\nabla V\cdot \nabla
}{V}+\frac{\epsilon ^{2}}{2}\nabla \left( \frac{\nabla V}{V}\right) .
\end{equation}%
Note that the quadratic form associated with $L_{\epsilon }-\tilde{L}%
_{\epsilon }$ is zero after the integration by parts. We establish
invertibility of $\tilde{L}_{\epsilon }$ on any element of $L^{2}(\mathbb{R}%
^{d})$ in the following lemma.

\begin{lemma}
\label{lemma-resolvent} Assume that $\nabla \log (V)\in L^{\infty }(\mathbb{R%
}^{d})$. There exists a positive constant $\epsilon _{0}$ such that for
every $\epsilon \in (0,\epsilon _{0})$ and for every $f\in L^{2}(\mathbb{R}%
^{d})$, the following is true:
\begin{equation}
\Vert \tilde{L}_{\epsilon }^{-1}f\Vert _{L^{2}}+\epsilon \Vert \nabla \tilde{%
L}_{\epsilon }^{-1}f\Vert _{L^{2}}\leq \Vert f\Vert _{L^{2}}.
\label{resolvent-1}
\end{equation}%
Additionally, if $\Delta \log (V)\in L^{\infty }(\mathbb{R}^{d})$, then for
every $f\in H^{1}(\mathbb{R}^{d})$, the following is true as well:
\begin{equation}
\Vert \tilde{L}_{\epsilon }^{-1}f\Vert _{H^{1}}\leq \Vert f\Vert _{H^{1}}.
\label{resolvent-2}
\end{equation}
\end{lemma}

\begin{proof}
Under the condition of $\nabla \log (V)\in L^{\infty }(\mathbb{R}^{d})$, the
last term in $\tilde{L}_{\epsilon }$ is a small bounded negative
perturbation to the first positive term, whereas the second term, $-\epsilon
^{2}\Delta $, is a nonnegative operator. The bilinear form
\begin{equation}
a(u,w):=\int_{\mathbb{R}^{d}}\left( 2\bar{w}u+\epsilon ^{2}\nabla \bar{w}%
\cdot \nabla u-\frac{\epsilon ^{2}|\nabla V|^{2}}{4V^{2}}\bar{w}u\right) dx
\end{equation}%
satisfies the boundedness and coercivity assumptions in the $H^{1}(\mathbb{R}%
^{d})$ space:
\begin{equation}
|a(u,w)| \leq C(V)\Vert u\Vert _{H^{1}}\Vert w\Vert _{H^{1}},\quad a(u,u)\geq
\epsilon ^{2}\Vert u\Vert _{H^{1}}^{2},
\end{equation}%
where the constant $C(V)>2$ depends on $\Vert \nabla \log (V)\Vert
_{L^{\infty }}$. By the Lax-Milgram Theorem, for every $f\in L^{2}(\mathbb{R}%
^{d})$, there is a unique $u\in H^{1}(\mathbb{R})$ such that
\begin{equation}
\Vert u\Vert _{L^{2}}^{2}+\epsilon ^{2}\Vert \nabla u\Vert _{L^{2}}^{2}\leq
a(u,u)=\int_{\mathbb{R}^{d}}\bar{u}fdx.
\end{equation}%
By the Cauchy--Schwarz inequality, we obtain the bounds (\ref{resolvent-1}).
Under the additional condition of $\Delta \log (V)\in L^{\infty }(\mathbb{R}%
^{d})$, we apply operator $\nabla $ to $\tilde{L}_{\epsilon }u=f$ and write
the corresponding equation in the weak form,
\begin{equation}
a(\nabla u,\nabla u)=\int_{\mathbb{R}^{d}}\nabla \bar{u}\cdot \nabla fdx+%
\frac{\epsilon ^{2}}{4}\int_{\mathbb{R}^{d}}\bar{u}\nabla u\cdot \nabla
\left( \frac{|\nabla V|^{2}}{V^{2}}\right) dx.
\end{equation}%
Again applying the Cauchy--Schwarz inequality and using the smallness of $%
\epsilon ^{2}$, we obtain the bound (\ref{resolvent-2}).
\end{proof}

Using Lemma \ref{lemma-resolvent}, the persistence problem (\ref{persistence}%
) can be rewritten as the fixed-point equation:
\begin{equation}
r=\Phi _{\epsilon }(r):=\epsilon ^{2}\tilde{L}_{\epsilon }^{-1}F+\tilde{L}%
_{\epsilon }^{-1}(\tilde{L}_{\epsilon }-L_{\epsilon })r+\tilde{L}_{\epsilon
}^{-1}N(r).  \label{fixed-point}
\end{equation}%
Using the contraction mapping principle, we prove the following result.

\begin{lemma}
\label{lemma-fixed-point} Assume that $\nabla \log (V),\Delta \log (V)\in
L^{\infty }(\mathbb{R}^{d})$ and $\nabla \log (V)\in H^{2}(\mathbb{R}^{d})$.
There exist positive constants $\epsilon _{0}$ and $C_{0}$ such that for
every $\epsilon \in (0,\epsilon _{0})$, there exists a unique solution $r\in
H^{1}(\mathbb{R}^{d})$ of the fixed-point equation (\ref{fixed-point})
satisfying
\begin{equation}
\Vert r \Vert _{H^{1}}\leq C_{0}\epsilon ^{2}.  \label{fixed-point solution}
\end{equation}
\end{lemma}

\begin{proof}
We will prove that under the assumptions of the theorem, operator $\Phi
_{\epsilon }$ is a contraction on the ball $B_{\delta }(H^{1}(\mathbb{R}%
^{d}))$ of radius $\delta $ if $\delta =\mathcal{O}(\epsilon ^{2})$ as $%
\epsilon \rightarrow 0$.

From the assumption of $\nabla \log (V)\in L^{\infty }(\mathbb{R}^{d})\cap
H^{2}(\mathbb{R}^{d})$, we realize that $F\in H^{1}(\mathbb{R}^{d})$.
Applying the bound (\ref{resolvent-2}), we obtain that, for $\epsilon >0$
sufficiently small, there is $C_{1}>0$ such that
\begin{equation}
\Vert \epsilon ^{2}\tilde{L}_{\epsilon }^{-1}F\Vert _{H^{1}}\leq \epsilon
^{2}C_{1}.
\end{equation}%
By Sobolev's embedding of $H^{1}(\mathbb{R}^{d})$ to $L^{p}(\mathbb{R}^{d})$
for any $p\geq 2$ if $d=1$, for $2\leq p<\infty $ if $d=2$, and $2\leq p\leq
6$ if $d=3$, and by the estimate (\ref{resolvent-1}), we obtain that, for $%
\epsilon >0$ sufficiently small, there is $C_{2}>0$ such that
\begin{equation}
\Vert \tilde{L}_{\epsilon }^{-1}N(r)\Vert _{H^{1}}\leq \epsilon ^{-1}\Vert
N(r)\Vert _{L^{2}}\leq \epsilon ^{-1}(3\Vert r\Vert _{L^{4}}^{2}+\Vert
r\Vert _{L^{6}}^{3})\leq C_{2}\epsilon ^{-1}(\Vert r\Vert _{H^{1}}^{2}+\Vert
r\Vert _{H^{1}}^{3}).
\end{equation}%
Finally, under the conditions of $\nabla \log (V),\Delta \log (V)\in
L^{\infty }(\mathbb{R}^{d})$, we have the bounds
\begin{equation}
\Vert (\tilde{L}_{\epsilon }-L_{\epsilon })u\Vert _{L^{2}}\leq \epsilon
^{2}\Vert \nabla \log (V)\Vert _{L^{\infty }}\Vert \nabla u\Vert _{L^{2}}+%
\frac{1}{2}\epsilon ^{2}\Vert \Delta \log (V)\Vert _{L^{\infty }}\Vert
u\Vert _{L^{2}},
\end{equation}%
hence, using estimate (\ref{resolvent-1}), we obtain that, for $\epsilon >0$
sufficiently small, there is $C_{3}>0$ such that
\begin{equation}
\Vert \tilde{L}_{\epsilon }^{-1}(\tilde{L}_{\epsilon }-L_{\epsilon })u\Vert
_{H^{1}}\leq \epsilon ^{-1}\Vert (\tilde{L}_{\epsilon }-L_{\epsilon })u\Vert
_{L^{2}}\leq \epsilon C_{3}\Vert u\Vert _{H^{1}}.
\end{equation}%
From these three estimates, it is clear that $\Phi _{\epsilon }$ maps a ball
$B_{\delta }(H^{1}(\mathbb{R}^{d}))$ of radius $\delta =C_{0}\epsilon ^{2}$
to itself, where $C_{0}>C_{1}$ independently of $\epsilon $. Similar
estimates on the Lipschitz continuous nonlinear term $N(r)$ and perturbation
operator $\tilde{L}_{\epsilon }^{-1}(\tilde{L}_{\epsilon }-L_{\epsilon })$
show that $\Phi _{\epsilon }$ is a contraction on the ball $B_{\delta
}(H^{1}(\mathbb{R}^{d}))$ of radius $\delta =C_{0}\epsilon ^{2}$. Hence, the
assertion of the theorem follows from the Banach fixed-point theorem.
\end{proof}

\begin{remark}
Sobolev's embedding of $H^s(\mathbb{R}^d)$ to $L^{\infty}(\mathbb{R}^d)$ for
$s > \frac{d}{2}$ allows us to replace the three conditions of the theorem
by only one condition: $\nabla \log(V) \in H^2(\mathbb{R})$ if $d = 1$ and $%
\nabla \log(V) \in H^3(\mathbb{R}^d)$ if $d = 2$ or $d = 3$. With this refinement,
Theorem \ref{theorem-fixed-point} follows from Lemma \ref{lemma-fixed-point} after
the decomposition $u = u_0(1 + r)$ is used.
\end{remark}

\begin{remark}
If the condition of the theorem are replaced by weaker conditions
\begin{equation}
\nabla \log (V),\Delta \log (V)\in L^{\infty }(\mathbb{R}^{d})\cap L^{2}(%
\mathbb{R}^{d}),
\end{equation}%
then, applying bound (\ref{resolvent-1}), we obtain
\begin{equation}
\Vert \epsilon ^{2}\tilde{L}_{\epsilon }^{-1}F\Vert _{H^{1}}\leq \epsilon
\Vert F\Vert _{L^{2}},
\end{equation}%
but radius $\delta =C_{0}\epsilon $ is critical for the contraction of
mapping $\Phi _{\epsilon }$ because of the quadratic term in $%
N(r)$. Hence, fixed-point arguments can only be used if $\Vert F\Vert _{L^{2}}$
is sufficiently small.
\end{remark}

In the end of this article, we discuss several examples.

\begin{itemize}
\item If $V$ grows algebraically at infinity with any rate $\alpha >0$, that is, if
$$
V(x)\sim |x|^{\alpha } \quad \mbox{\rm as} \quad |x|\rightarrow \infty
$$
(the nonlinear confinement of this kind was adopted in Ref. \cite{Barcelona1}), then
$$
|\nabla \log (V)|\sim |x|^{-1} \quad \mbox{\rm and} \quad \Delta \log (V)|\sim |x|^{-2}.
$$
Assuming smoothness of $V$, these conditions provide $F\in H^{1}(\mathbb{R}%
^{d})$ if $d=1,2,3$, hence Theorem \ref{theorem-fixed-point} holds for such potentials for any $%
\alpha >0$. (Of course, $u\in L^{2}(\mathbb{R}^{d})$ if and only if $\alpha >d$).
Some exact expressions are available for particular $V$ and $\epsilon$ \cite{Barcelona4}.

\item If $V$ grows like the exponential or Gaussian funciton (such as in
the models introduced in Refs. \cite{Barcelona2,Barcelona5}),
then the assumption $F\in H^{1}(\mathbb{R}^{d})$ fails for any $d=1,2,3$.
Nevertheless, if $V = (1 + \beta |x|^2) e^{\alpha |x|^2}$ with $\alpha, \beta > 0$, then
the analytic expression is available \cite{Barcelona2} for a particular value of $\epsilon = \epsilon_0$:
\begin{equation}
\label{exact-solution}
u = \frac{\epsilon \alpha}{\sqrt{\beta}} e^{-\frac{\alpha}{2}|x|^2},  \quad
\epsilon_0 = \frac{\sqrt{\beta}}{\sqrt{\alpha^2 + d \alpha \beta}}.
\end{equation}
However, because $F \notin H^1(\mathbb{R}^d)$ ($F$ is not even bounded at infinity),
it is not clear if there exists a family of stationary states for any $\epsilon \in (0,\epsilon_0)$
that connects the TF approximation (\ref{TF-approximation}) and the exact solution (\ref{exact-solution}).

\item If $V$ is a symmetric double-well potential, then Theorem \ref%
{theorem-fixed-point} justifies the construction of a symmetric
stationary state $u$. Symmetry-breaking bifurcation may happen in
double-well potentials, but it cannot happen to the symmetric state due to
uniqueness arguments. Therefore, such a bifurcation may only happen to an
anti-symmetric stationary state.
\end{itemize}

In conclusion, we have produced a rigorous proof of the proximity of the
self-trapped states, produced by the TF\ (Thomas-Fermi) approximation in the
recently developed models with the spatially growing local strength of the
defocusing cubic nonlinearity, to the true ground state, in the space of any
dimension, $d=1,2,3$. As an extension of the analysis, it may be interesting
to justify the empiric use of the TF approximation for the description of
self-trapped modes with intrinsic vorticity (by themselves, they are not
ground states, but may play such a role in the respective reduced radial
models \cite{Barcelona3,Barcelona5,IM}). Another relevant extension can be developed for
two-component models with the nonlinear confinement of the same type \cite%
{Barcelona-2comp}.

\vspace{0.5cm}

{\bf Acknowledgments:} The work of D.P. is supported by the Ministry of Education
and Science of Russian Federation (the base part of the state task No. 2014/133).


\begin{thebibliography}{99}
\bibitem{KA} Y. S. Kivshar and G. P. Agrawal, \textit{Optical Solitons: From
Fibers to Photonic Crystals }(Academic Press: San Diego, 2003).

\bibitem{DP} T. Dauxois and M. Peyrard, \textit{Physics of Solitons} (Cambridge
University Press: New York, 2006).

\bibitem{Barcelona1} O. V. Borovkova, Y. V. Kartashov, B. A. Malomed, and L.
Torner, ``Algebraic bright and vortex solitons in defocusing media", Opt. Lett. \textbf{36}, 3088 (2011).

\bibitem{Barcelona2} O. V. Borovkova, Y. V. Kartashov, L. Torner, and B. A.
Malomed, ``Bright solitons from defocusing nonlinearities",
Phys. Rev. E \textbf{84}, 035602(R) (2011).

\bibitem{Barcelona-2comp} Y. V. Kartashov, V. A. Vysloukh, L. Torner, and B.
A. Malomed, ``Self-trapping and splitting of bright vector solitons under inhomogeneous defocusing nonlinearities",
Opt. Lett. \textbf{36}, 4587 (2011).

\bibitem{Barcelona3} Q. Tian, L. Wu, Y. Zhang, and J.-F. Zhang,
``Vortex solitons in defocusing media with spatially inhomogeneous nonlinearity",
Phys. Rev. E \textbf{85}, 056603 (2012).

\bibitem{Barcelona4} Y. Wu, Q. Xie, H. Zhong, L. Wen, and W. Hai,
``Algebraic bright and vortex solitons in self-defocusing media with spatially inhomogeneous nonlinearity",
Phys. Rev. A \textbf{87}, 055801 (2013).

\bibitem{Barcelona5} R. Driben, Y. V. Kartashov, B. A. Malomed, T. Meier,
and L. Torner, ``Soliton gyroscopes in media with spatially growing repulsive nonlinearity",
Phys. Rev. Lett. \textbf{112}, 020404 (2014).


\bibitem{Thomas} L.H. Thomas, ``The calculation of atomic fields", Proc. Cambridge Philos. Soc. {\bf 23} (1927), 542.
\bibitem{Fermi} E. Fermi, ``Statistical method of investigating electrons in atoms", Z. Phys. {\bf 48} (1928), 73--79.

\bibitem{IM} R. Ignat and V. Millot,
``Energy expansion and vortex location for a two-dimensional rotating Bose-Einstein condensate",
Rev. Math. Phys.  {\bf 18}  (2006),  no. 2, 119--162.

\bibitem{GP} C. Gallo and D. Pelinovsky,
``On the Thomas-Fermi ground state in a harmonic potential",
 Asymptot. Anal.  {\bf 73}  (2011),  no. 1-2, 53--96.

\bibitem{KS} G. Karali and C. Sourdis, ``The ground state of a Gross--Pitaevskii energy with general
potential in the Thomas--Fermi limit", Arch. Rat. Mech. Appl., in print (2014), arXiv:1205.5997
\end{thebibliography}
\end{document}